\documentclass[runningheads]{llncs}
\usepackage[T1]{fontenc}
\usepackage{graphicx}
\usepackage{comment}

\usepackage{algorithm}
\usepackage[noend]{algpseudocode}

\usepackage[T1]{fontenc}
\usepackage[utf8]{inputenc}
\usepackage{lmodern}

\usepackage{hyperref}

\usepackage{amsmath}
\usepackage{cite}
\usepackage{cleveref}

\usepackage{amssymb}
\usepackage{amsfonts}
\usepackage{mathrsfs}
\usepackage{wrapfig}

\usepackage{mathtools} % write text [under]{over} arrow $\xrightarrow[\text{world}]{\text{hello}}$

\newtheorem{observation}{Observation}

\usepackage{appendix}

\usepackage{lineno}
\begin{document}
\title{Round-Delayed Amnesiac Flooding\thanks{A preliminary version of this paper appeared at the Proceedings of ALGOWIN 2025, pages 32-45 (\textbf{Best paper award})~\cite{AlafinMS25}.}}
\author{Oluwatobi Alafin\inst{1}%\orcidID{} 
\and
George B. Mertzios\inst{2}\thanks{Supported by the EPSRC grant EP/P020372/1.}\orcidID{0000-0001-7182-585X} 
\and
Paul G. Spirakis\inst{1}\thanks{Supported by the EPSRC grant EP/P02002X/1.}\orcidID{0000-0001-5396-3749}}
\authorrunning{Oluwatobi Alafin, George B. Mertzios, and Paul G. Spirakis}
\institute{Department of Computer Science, University of Liverpool, UK \email{o.f.alafin@liverpool.ac.uk}, \ \email{p.spirakis@liverpool.ac.uk}\and
Department of Computer Science, Durham University, UK \email{george.mertzios@durham.ac.uk}}
\maketitle              % typeset the header of the contribution
\begin{abstract}
We present a comprehensive analysis of Round-Delayed Amnesiac Flooding (RDAF), a variant of Amnesiac Flooding that introduces round-based asynchrony through adversarial delays. We establish fundamental properties of RDAF, including termination characteristics for different graph types and decidability results under various adversarial models. Our key contributions include: (1) a formal model of RDAF incorporating round-based asynchrony, (2) a proof that flooding always terminates on acyclic graphs despite adversarial delays, (3) a construction showing non-termination is possible on any cyclic graph, (4) a demonstration that termination is undecidable with arbitrary computable adversaries, and (5) the introduction of Eventually Periodic Adversaries (EPA) under which termination becomes decidable. These results enhance our understanding of flooding in communication-delay settings and provide insights for designing robust distributed protocols.

% \vspace{0.2cm}

% \noindent \textbf {This paper is eligible for the best student paper award. Due to space constraints, some proofs and details are deferred to a clearly marked appendix to be read at the discretion of the program committee.}

\end{abstract}
\section{Introduction}\label{intro-sec}

Flooding algorithms \cite{Aspnes2019} serve as fundamental primitives in distributed computing for information dissemination, with applications ranging from network discovery to emergency broadcast systems. While traditional flooding maintains message histories to prevent redundant transmissions \cite{AttiyaWelch2004}, such approaches become impractical in resource-constrained environments like sensor networks, IoT devices, or networks with high churn rates where maintaining consistent state is challenging.

Amnesiac Flooding (AF) addresses these limitations by eliminating message history, requiring nodes to make forwarding decisions based solely on current information \cite{hussak2019termination}. However, existing AF analyses assume perfect synchrony—an unrealistic assumption in practical networks where delays, failures, and asynchrony are the norm rather than the exception.

This paper introduces \textit{Round-Delayed Amnesiac Flooding (RDAF)}, which bridges the gap between theoretical AF models and practical network conditions. RDAF maintains the memory-efficiency of amnesiac approaches while incorporating realistic asynchronous behaviour through adversarial delays; thus, our model can be also termed \textit{Round-Delayed Amnesiac Flooding (RDAF)}.
Our key insight is that even with minimal state and adverse conditions, we can characterise precise conditions under which flooding terminates, providing both positive results (guaranteed termination in acyclic networks) and fundamental limitations (undecidability with arbitrary adversaries).

The significance of our results extends beyond flooding protocols. By establishing when termination analysis becomes undecidable and identifying restricted adversary models (Eventually Periodic Adversaries) where it remains decidable, we contribute to the broader understanding of computability limits in asynchronous distributed systems.

While our model makes specific assumptions (such as nodes detecting blocked edges), these capture realistic scenarios and enable rigorous analysis of fundamental limits. The dichotomy between acyclic and cyclic graphs, and between arbitrary and periodic adversaries, reveals deep structural properties that inform the design of practical flooding protocols.

\subsection{Model Context and Novelty}

RDAF occupies a unique position in the spectrum of distributed system models. Unlike classical asynchronous models that allow arbitrary message delays and reorderings\cite{AttiyaWelch2004}, or partially synchronous models that impose eventual bounds on communication delays\cite{DworkLynch1988}, RDAF maintains a synchronous round structure while allowing adversarial edge-level asynchrony within rounds.

\subsubsection*{Adversarial Model Justification}
A key aspect of our model is that nodes are aware of which outgoing edges are currently delayed by the adversary. While this may initially seem like a strong assumption, it captures several practical scenarios:
\begin{itemize}
    \item \textbf{Failed transmission detection}: In many network protocols, nodes receive acknowledgments or can detect transmission failures through timeout mechanisms or carrier sensing.
    \item \textbf{Scheduled maintenance}: In managed networks, nodes may be informed of temporary link unavailability due to scheduled maintenance or known congestion patterns.
    \item \textbf{Visible network conditions}: In wireless networks, nodes can often detect poor channel conditions or interference that prevents successful transmission.
\end{itemize}

This modeling choice allows us to study the fundamental limits of flooding under adversarial conditions while maintaining some feedback about the network state. Alternative models where nodes lack this information would require additional mechanisms (such as acknowledgments or timeouts) that would fundamentally change the nature of the flooding protocol.

\subsubsection*{Memory Model and Amnesiac Nature}
The term ``amnesiac'' in our context requires careful interpretation. Traditional amnesiac flooding assumes that nodes retain \emph{no state} between rounds. Our variant relaxes this to what we call \emph{structured amnesia}: nodes forget all message history but maintain a bounded amount of state (destination sets) that is recomputed based on current round information. Specifically:
\begin{itemize}
    \item Nodes do not remember which nodes they have received messages from in previous rounds.
    \item Nodes only maintain destination sets that are \emph{functionally determined} by the current round's receipts and delays.
\end{itemize}

This structured amnesia is motivated by resource-constrained environments where maintaining full message history is infeasible, but nodes can afford $O(|N(v)|)$ memory for immediate forwarding decisions, where $N(v)$ denotes the set of neighbours of node $v$. This represents a middle ground between full amnesia and traditional flooding.

\subsection{Our Contributions}

This paper makes several significant contributions:

\begin{enumerate}
    \item \textbf{Formal Model}: A comprehensive mathematical framework for RDAF, including precise definitions for system state, adversarial delay functions, and termination conditions.

    \item \textbf{Termination Analysis}: Proof that RDAF always terminates on acyclic graphs with provable bounds, and demonstration that any cyclic graph admits non-termination.

    \item \textbf{Decidability Results}: Proof of undecidability for arbitrary computable adversaries via reduction from the halting problem, and introduction of the Eventually Periodic Adversary (EPA) model under which termination becomes decidable.
\end{enumerate}

\subsection{Organisation}

Section \ref{sec-related} discusses related work. Section \ref{sec-prelim} presents our formal model and defines key properties. Section \ref{sec-term-dichotomy} demonstrates non-termination in cyclic graphs. Section~\ref{sec:PIS} develops the theory of periodic infinite schedules, providing a framework for analysing recurrent behaviour. Section \ref{sec-undecidability} proves undecidability for arbitrary computable adversaries. Section \ref{sec:EPA} introduces Eventually Periodic Adversaries (EPA), establishes decidability and provides complexity bounds. 
% In the full version of the paper we show that all non-terminating schedules in the EPA model are periodic infinite schedules. 
% The appendix also provides detailed proofs, worked examples demonstrating schedule evolution, and extensions of our treatment of recurrent behaviour to zero-density and generalised periodic schedules.

% \subsection{Technical Preliminaries}

% Throughout, we work with simple, finite, connected graphs $G = (V,E)$ with node set $V$ and edge set $E$. We assume a single source node $g_0 \in V$ with initial message $m_0$. 
% For any graph $G$ and nodes $u,v \in V$, we denote by $\text{dist}(G)(u,v)$ the length of the shortest path between $u$ and $v$ in $G$.

\section{Related Work}\label{sec-related}

Amnesiac Flooding originated with Hussak and Trehan's work~\cite{HussakTrehan2020b}, establishing fundamental properties in synchronous settings. Their analysis proved termination bounds---exactly $e$ rounds for bipartite and between $e$ and $e+d+1$ rounds for non-bipartite graphs (where $e$ is source eccentricity, $d$ is graph diameter)---demonstrating AF's asymptotic time optimality against the $\Omega(d)$ broadcast lower bound.

Turau~\cite{turau2020analysis} revealed deeper complexity aspects through the $(k,c)$-flooding problem: finding $k$ nodes to guarantee termination within $c$ rounds under concurrent flooding. Its NP-completeness highlighted inherent optimisation challenges. Sharp bounds showed significant disparities between bipartite and non-bipartite graphs, introducing a behaviour-preserving construction mapping between them.

Hussak and Trehan extended their analysis~\cite{HussakTrehan2023} to multi-source scenarios, proving $e(I)$-round termination for $I$-bipartite graphs with source set $I$. Their fixed-delay analysis showed termination by round $2d+\tau-1$ for single-edge delays of duration $\tau$ in bipartite graphs and established termination for multiple-edge fixed delays in cycles.

Bayramzadeh et al.~\cite{BayramzadehEtAl2021} proved termination for multiple-message AF in the unranked full-send case, previously conjectured non-terminating, showing $D\cdot(2k-1)$ rounds for bipartite and $(2D+1)\cdot(2k-1)$ rounds for non-bipartite graphs (a global total of $k$ messages). Their introduction of graph diameter knowledge as a parameter suggested new model variants.

\textbf{Comparison with Asynchronous Flooding Models.}
Hussak and Trehan~\cite{HussakTrehan2020b} briefly discuss an asynchronous variation of amnesiac flooding, and they demonstrate with a small example that this variation does not guarantee termination, in contrast to their synchronous model. 
In this %asynchronous message passing 
model of \cite{HussakTrehan2020b}, 
%the computation still proceeds in global synchronous rounds (similarly to our model), but 
the adversary can decide a delay of message delivery on any link. Once a node sends a message, this message will definitely be delivered at some future round. 
In contrast, in our model, if a message from node $u$ to node $v$ is delayed by the adversary, the initiator node $u$ keeps a note of this and tries to re-send the message to $v$ again and again, until either (i) the message is delivered to $v$, or (ii) $u$ receives the message from $v$, in which case $u$ stops trying to send the message to $v$.

The possibility of a message not being delivered, due to the last case, makes non-termination in our model much less trivial, compared to \cite{HussakTrehan2020b}. 
We comprehensively investigate this model and we provide a periodic-schedule normal form for it: every infinite execution can be compressed into an ultimately periodic delay pattern, which in turn enables decidability and undecidability results. 
% More specifically, termination is \emph{undecidable} in the full adaptive model, yet becomes decidable under our \emph{eventually-periodic adversary} (EPA) restriction. 
Summarizing, our model is complementary to the model of \cite{HussakTrehan2020b}, and our work provides a rigorous framework that delineates the exact boundary between terminating and non-terminating behaviour.
% which complements the illustrative constructions of \cite{HussakTrehan2020b} in their model.
% extensions turn the illustrative constructions of \cite{HussakTrehan2020b} into a rigorous framework and delineate the exact boundary between terminating and non-terminating behaviour.

\textbf{Comparison with Stateless Flooding Approaches.}
The stateless flooding algorithm by Adamek et al.~\cite{adamek2017} achieves statelessness through a different mechanism than our structured amnesia. Their algorithm uses send queues where messages are discarded upon encountering ``mates''—pairs of messages with swapped sender/receiver addresses. Crucially, their model assumes fair scheduling where every queued message is eventually transmitted or removed, without adversarial interference. This synchronous assumption fundamentally differs from our round-delayed model where an adversary controls edge availability. 

While Adamek et al. prove termination under fair scheduling, we establish when termination remains decidable despite adversarial delays (EPA model) or becomes undecidable (arbitrary computable adversaries). Our structured amnesia—recomputing destination sets based on current round information—provides a framework for analysing flooding under hostile scheduling conditions that previous stateless approaches did not consider.

\section{Model, Preliminaries and Notation}\label{sec-prelim}

In this section, we present the computational model for round-based asynchronous systems, 
and we describe the RDAF protocol that operates within this model. 
We summarize all our notation in Table~\ref{notation-table} below, while the individual symbols are explained in detail later in this section.

\begin{table}[h]
\centering
\caption{Summary of Notation and Terminology}
\begin{tabular}{|l|l|p{0.5\textwidth}|}
\hline
\textbf{Symbol} & \textbf{Type} & \textbf{Description} \\
\hline
$G = (V,E)$ & Graph & Simple, finite, connected graph with vertices $V$, edges $E$ \\
\hline
% $\text{dist}(G,u,v)$ & Natural & Length of the shortest path between nodes $u$ and $v$ in $G$\\
% \hline
$g_0$ & Node & Source node initiating flooding \\
\hline
$d$ & Function & $d: \mathbb{N} \times V \times E \to \{0,1\}$ adversary's delay function \\
\hline
$M(j,u)$ & Boolean & Message indicator for node $u$ at round $j$ \\
\hline
$s(j,u)$ & Set & Sources delivering to node $u$ in round $j$ \\
\hline
$\text{dest}(j,u)$ & Set & Destinations for node $u$ in round $j+1$ \\
\hline
$S(j,u)$ & Triple & State record $(M(j,u), s(j,u), \text{dest}(j,u))$ \\
\hline
$r(j)$ & Set & Subset of $V \times E$ representing transmissions in round $j$ \\
\hline
$\sigma(j)$ & Pair & Complete system configuration $[(S(j, \cdot), r(j))]$ at round $j$ (Definition~\ref{def:configuration}) \\
\hline
$c$ & Natural & Stabilisation round for periodic delay functions (Definition~\ref{def:periodic-delay}) / periodic infinite schedules (Definition~\ref{def:PIS})) \\
\hline
$l$ & Natural & Cycle length for periodic delay functions (Definition~\ref{def:periodic-delay}) \\
\hline
%%% $T(n)$ & Ratio & Transmission density up to round $n$ (Definition \ref{def:density}) \\
%%% \hline
\end{tabular}
\label{notation-table}
\end{table}

\subsection{Computational Model}

\textbf{Network Structure.} We consider a simple, finite, connected graph $G = (V,E)$ with a distinguished source node $g_0 \in V$ possessing initial message $m_0$. For every node $v$ we denote by $N(v)$ the set of neighbours of $v$.

\textbf{Round-Based Asynchrony.} The system proceeds in synchronous rounds, but an adversary can selectively make edges unavailable for message transmission. Formally:

\begin{definition}[Adversarial Delay Function]
A delay function $d: \mathbb{N} \times V \times E \to \{0,1\}$ specifies for each round $j$, node $v$, and incident edge $e = \{v,u\}$ whether transmission from $v$ along $e$ is blocked ($d(j,v,e) = 1$) or allowed ($d(j,v,e) = 0$).
\end{definition}

\begin{definition}[Finite Delay Property]
A delay function $d$ satisfies the finite delay property if for every node-edge pair $(v,e)$, any sequence of consecutive rounds where $d(j,v,e) = 1$ is finite. Formally, for all $v \in V$ and incident edges $e$, if $d(j,v,e) = 1$ for $j \in [t_1, t_2]$, then $t_2 - t_1$ is finite.
\end{definition}

\textbf{Key Model Assumption.} Nodes are aware of which of their incident edges are currently blocked. This models scenarios where transmission failures are detectable (e.g., through carrier sensing, acknowledgments, or network management protocols).

\subsection{The RDAF Protocol}

Within the above model, we define the Round-Delayed Amnesiac Flooding protocol.

\textbf{Node State.} Each node $v$ maintains:
\begin{itemize}
    \item $M(j, v) \in \{0,1\}$: whether $v$ possesses the message in round $j$.
    \item $s(j, v) \subseteq V$: source set - neighbours that successfully delivered the message to $v$ in round $j$.
    \item $\text{dest}(j, v) \subseteq V$: destination set - neighbours to which $v$ intends to forward the message.
\end{itemize}

\textbf{Structured Amnesia.} The protocol is called ``amnesiac'' because:
\begin{itemize}
    % \item Nodes do not maintain message identifiers or history
    \item Nodes do not maintain history: i.e.,~nodes don't remember which nodes they received messages from (or sent messages to) in previous rounds.
    \item The destination set is functionally determined by recent receptions and current delays.
    \item When destination sets empty and no delayed transmissions remain, nodes return to their initial state.
\end{itemize}

The key insight is that while nodes maintain destination sets across rounds (not strictly amnesiac), this state is \emph{recomputed} based on current information rather than accumulated history. When a node receives the message from new sources, it completely recomputes its forwarding strategy.
\medskip

\textbf{Protocol Operation.} Each round $j$ proceeds as follows:

\begin{enumerate}
\item \textbf{Delay Phase}: The adversary specifies $d(j,v,e)$ for all node-edge pairs.
\item \textbf{Transmission Phase}: Each node $v$ with $M(j-1,v) = 1$ attempts to transmit to all $u \in \text{dest}(j-1,v)$ where $d(j,v,\{v,u\}) = 0$.
\item \textbf{Reception Phase}: Nodes receive messages from successful transmissions.
\item \textbf{State Update Phase}: Nodes update their state according to the following rules:
\end{enumerate}

\textbf{State Update Rules.} For node $u$ transitioning from round $j$ to $j+1$:

\textbf{Message Possession:}
\[ M(j+1, u) = \begin{cases}
1, & \text{if } u \text{ has pending delayed transmissions from round } j, \\
1, & \text{if } u \text{ receives the message in round } j+1, \\
0, & \text{otherwise}.
\end{cases} \]

\textbf{Source Set:} 
\[ s(j+1, u) = \{v \in V : \{v,u\} \in E, u \in \text{dest}(j, v), d(j+1, v, \{v,u\}) = 0\}. \]

\textbf{Destination Set:}
The update rule for destination sets captures the ``structured amnesia'':
\[ \text{dest}(j+1, u) = \begin{cases}
\{v \in \text{dest}(j, u) : d(j+1, u, \{u,v\}) = 1\}, & \text{if continuing delayed transmission,} \\
N(u) \setminus s(j+1, u), & \text{if newly receiving message,} \\
\emptyset, & \text{if no message possessed.}
\end{cases} \]

\begin{observation}[Persistence of Destination Sets (PDS)]
\label{obs:pds} Let $\{u, v\} \in E$. If $u$ receives in round $j$ from $w\neq v$ and does not receive from $v$ in the same round, then $v$ belongs to $\text{dest}(u)$ until $u$ delivers to $v$ or receives from $v$.

% \begin{proof}[Proof Sketch]
% Destination sets only change on successful transmission or message receipt, neither of which occurs by assumption. See \hyperref[lem:pds-app]{Appendix \ref{app:proofs}} for the complete proof.
% \end{proof}
\end{observation}

A key property of the protocol is that when $u$ receives the message from new sources, it \emph{recomputes} its destination set as all neighbours except those that just delivered the message, effectively ``forgetting'' its previous forwarding intentions.

\subsection{System Evolution and Termination}

\textbf{State Function.} The system state is captured by $S: \mathbb{N} \times V \to \text{StateRecord}$ where $S(j,u) = (M(j,u), s(j,u), \text{dest}(j,u))$.

\textbf{Round Function.} The transmission function $r(j)$ records actual message transmissions in round $j$:
\[ r(j) = \{(v, \{v,w\}) : v \in V, w \in \text{dest}(j-1, v), d(j, v, \{v,w\}) = 0\}. \]

\medskip
\noindent\textbf{State Update Function.}
\label{def:state-update}
The state update function defines how node states evolve. For node $v$ in round $j$, computing the next state requires:
\begin{itemize}
    \item Current graph state $S(j,\cdot): V \to \text{StateRecord}$.
    \item Current delay decisions $d(j,\cdot,\cdot): V \times E \to \{0,1\}$.
    \item Node's local state $S(j,v)$ and incident delays $d(j,v,\cdot)$.
\end{itemize}

While updates depend on graph-wide state and delays, these parameters remain fixed when computing individual node updates in a given round. Thus, we can express the state update as:

\(S(j+1, v) := u(v, S(j, v), d(j, v, \cdot))\).

where $u$ implicitly references the global state $S(j,\cdot)$ and delays $d(j,\cdot,\cdot)$ fixed for round $j$.

\textbf{Initial Configuration.} At round zero:
\begin{itemize}
    \item Source: $S(0, g_0) = (1, \emptyset, N(g_0))$.
    \item Others: $S(0, u) = (0, \emptyset, \emptyset)$ for $u \neq g_0$.
\end{itemize}

\medskip
\noindent\textbf{Termination.} 
\label{def:termination}
We say that flooding has terminated ``by'' (at or before) round $t\in \mathbb{N}$ if $M(t, v) = 0$ for every $v\in V$, i.e.,~no node $v$ has the message at round $t$. Clearly, this is equivalent with saying that $M(j, v) = 0$ for every $j\geq t$ and for every $v \in V$, i.e.,~if flooding has terminated by round $t$ then no node has the message in any round after round $t$. 

\begin{definition}%[Termination Round]
\label{def:termination-round}
The termination round $t_{\text{min}}$ is defined as:
\[ t_{\text{min}} = \min\{t \in \mathbb{N} : M(t, u) = 0, \text{ for every } u \in V\}. \]
% \[ t_{\text{min}} = \min\{t \in \mathbb{N} : \forall k > t : r(k) = \emptyset \land \forall u \in V: M(k, u) = 0\}. \]
\end{definition}

% $\forall k > t : r(k) = \emptyset \iff \forall u \in V: M(k, u) = 0 \iff dest(k, u) = \emptyset\}$

\begin{observation}[Persistence of Empty Destination Sets (PEDS)]
\label{obs:PEDS}
If a node's destination set is empty at round $t$, it remains empty in all subsequent rounds until the node receives the message.
\end{observation}

\begin{lemma}[Empty Destination Sets and Termination]
\label{lemma:empty-dest-term}
All destination sets are empty at round $t$ if and only if flooding has terminated by round $t$.
\end{lemma}

\begin{proof}
($\Rightarrow$) Let $t$ be a round where $\text{dest}(t, v) = \emptyset$ for all $v \in V$. By Observation~\ref{obs:PEDS}, $\text{dest}(j, v) = \emptyset$ for all $j > t$, all $v \in V$, unless some node receives a message. However, receiving a message requires a non-empty destination set in some previous round, contradicting our assumption about $t$. Therefore, no transmissions can occur after round $t$.

($\Leftarrow$) If flooding has terminated by round $t$, then for all $j > t$, $r(j) = \emptyset$ and $M(j, u) = 0$. Suppose for contradiction that at round $t$, some node $v$ has $\text{dest}(t, v) \neq \emptyset$. By the finite delay property, node $v$ must eventually be allowed to transmit, contradicting termination at round $t$.
\qed\end{proof}

% \begin{corollary}
% \label{termination-round}
% The earliest round at which all destination sets are empty is $t_{\text{min}}$.
% \end{corollary}

% Add this after the State Update Rules in Section 3

\section{Termination Dichotomy}
\label{sec-term-dichotomy}
% \subsection{Termination in Acyclic Graphs}

It is not hard to establish that RDAF always terminates on acyclic graphs, regardless of the adversarial strategy. 
In the remainder of this section we focus on graphs that contain at least one cycle, where we prove that every such graph admits a non-terminating strategy under RDAF by constructing a periodic infinite schedule (Definition \ref{def:PIS})\footnote{When restricting schedule information only to the cycle.} that maintains message circulation within the cycle.

%\subsection{Non-termination in Cyclic Graphs}

% We prove that any graph containing a cycle admits a non-terminating strategy under RDAF by constructing a periodic infinite schedule (Definition \ref{def:PIS})\footnote{When restricting schedule information only to the cycle.} that maintains message circulation within the cycle.

Let $G = (V,E)$ be an arbitrary graph containing a cycle $C = (V_C, E_C)$ where $V_C = \{v_1, \ldots, v_n\}$ and $E_C = \{\{v_i, v_{i+1 \bmod n}\} : 1 \leq i \leq n\}$. Let $c$ be the earliest round where a node in $V_C$ receives the message. Among nodes receiving the message in round $c$, designate one as $v_1$ and number remaining cycle nodes sequentially\footnote{We require that $v_2 \neq g_0$. This is because we assume that $v_2 \notin s(c, v_1)$, and this condition would not be satisfied if $v_2 = g_0$}.
We define the following delay function $d$:
\[
d(j, u, e) = \begin{cases}
0, & \text{if } j - c \equiv i \pmod{n} \text{ and } e = \{v_i, v_{i+1 \bmod n}\}, \\
1, & \text{otherwise}.
\end{cases}
\]

\begin{property}[Cyclic Propagation Pattern]
\label{property:cyclic-propagation}
If, for every $i \in \mathbb{N}^{+}$, we have that node $v_{i \bmod n}$ of the cycle $C$ transmits to node $v_{(i+1) \bmod n}$ of the cycle $C$ at round $c + i$, then we say that the \emph{cyclic propagation pattern} is satisfied for cycle $C$.
\end{property}

\begin{lemma}%%%[Cyclic Propagation]
\label{lem:cyclic-propagation}
The delay function $d$ ensures cyclic propagation.
\end{lemma}

\begin{proof}
By induction on $i$:

Base ($i = 1$): By definition of $c$, $v_1$ has message. Adversary allows transmission to $v_2$, which occurs as $v_2$ is in $v_1$'s destination set.

Step: Assume the hypothesis holds for some $k \geq 1$. In round $c + (k \bmod n)$:
\begin{itemize}
    \item $v_{k \bmod n}$ transmits to $v_{(k+1) \bmod n}$.
    \item $v_{(k+1) \bmod n}$ adds $v_{(k+2) \bmod n}$ to destination set.
    \item In round $c + ((k+1) \bmod n)$, transmission occurs by adversary strategy.
\end{itemize}
\qed\end{proof}

\begin{lemma}[Characterisation of Cyclic Pattern Disruption]
\label{lem:pattern-disruption}
The cyclic propagation pattern (Property \ref{property:cyclic-propagation}) is disrupted if and only if for some $k$, node $v_{k+1 \bmod n}$ transmits to $v_{k \bmod n}$ while $v_{k \bmod n}$ has the message but before $v_{k \bmod n}$ transmits to $v_{k+1 \bmod n}$.
\end{lemma}

\begin{proof}
($\Rightarrow$) Suppose the cyclic pattern is disrupted. By Observation \ref{obs:pds}, upon receipt of the message, each node $v_{k \bmod n}$ in the cycle adds and retains the next node $v_{k+1 \bmod n}$ in its destination set until successful transmission or it receives from $v_{k+1 \bmod n}$. Successful transmission will fulfil the cyclic propagation property (Property \ref{property:cyclic-propagation}). Therefore, the only way to prevent a scheduled transmission according to the cyclic pattern is for $v_{k+1 \bmod n}$ to transmit to $v_{k \bmod n}$ while the latter has the message but before its cyclic transmission occurs.

($\Leftarrow$) If in some round $j$ where $v_{k \bmod n}$ has the message and is awaiting its transmission window, $v_{k+1 \bmod n}$ transmits to $v_{k \bmod n}$, this removes $v_{k+1 \bmod n}$ from $v_{k \bmod n}$'s destination set. By cyclic propagation (Property \ref{property:cyclic-propagation}), $v_{k \bmod n}$ must transmit to $v_{k+1 \bmod n}$ in round $h$ where $h-c \equiv k \pmod n$. However, with $v_{k+1 \bmod n}$ removed from its destination set, this transmission cannot occur, violating the cyclic propagation property.
\qed\end{proof}

\begin{lemma}[Non-disruption of Cyclic Pattern]
\label{lem:non-disruption}
The cyclic propagation pattern cannot be disrupted by message flow in the reverse direction.
\end{lemma}

\begin{proof}
By Lemma \ref{lem:pattern-disruption}, disruption occurs if and only if $v_{k+1 \bmod n}$ transmits to $v_{k \bmod n}$ while $v_{k \bmod n}$ has the message but before $v_{k \bmod n}$ transmits to $v_{k+1 \bmod n}$. 

However, by construction of the delay function $d$, $v_{k+1 \bmod n}$ is permitted to transmit to $v_{k \bmod n}$ only in round $h$ where $h-c \equiv k \pmod n$. This is precisely the round where $v_{k \bmod n}$ transmits to $v_{k+1 \bmod n}$ according to the cyclic pattern. Therefore any reverse transmission must occur simultaneously with the cycle-preserving forward transmission, preventing the disruption characterised in Lemma \ref{lem:pattern-disruption}.
\qed\end{proof}

\begin{lemma}[External Message Preservation]
\label{lem:external-preservation}
Receiving messages from nodes outside the cycle does not disrupt cyclic propagation.
\end{lemma}

\begin{proof}
Consider any round $k > c$, node $v_i \in V_C$, and $u \in V \setminus V_C$ where $\{u, v_i\} \in E$.

Case $k - c \equiv i \pmod n$: Node $v_{i \bmod n}$ sends to $v_{i+1 \bmod n}$ according to cyclic propagation. By Lemma \ref{lem:pattern-disruption}, the only way this transmission could be disrupted is if $v_{i+1 \bmod n}$ transmits to $v_{i \bmod n}$ before $v_{i \bmod n}$'s scheduled transmission. However, by Lemma \ref{lem:non-disruption}, any such reverse transmission must occur simultaneously with the cycle-preserving forward transmission. Therefore, even if $v_{i+1 \bmod n}$ has received the message from a node outside the cycle, the cyclic propagation pattern continues unaffected.

Case $k - c \not\equiv i \pmod n$:
\begin{itemize}
    \item Transmission to $v_{i+1 \bmod n}$ delayed.
    \item By Observation \ref{obs:pds}, $v_{i+1 \bmod n}$ remains in destination set until successful delivery.
    \item By construction of the delay function, $v_{i+1 \bmod n}$ is not allowed to send to $v_{i \bmod n}$ until some future round $h: h - c \equiv i \pmod n$ returning to the former case.
\end{itemize}
\qed\end{proof}

\begin{theorem}[Cyclic Non-termination]
\label{thm:cyclic-nonterm}
For any graph $G = (V,E)$ containing a cycle, there exists a valid delay function $d$ such that flooding does not terminate.
\end{theorem}

\begin{proof}
By Lemmas \ref{lem:cyclic-propagation} and \ref{lem:external-preservation}, the delay function $d$ ensures cyclic propagation is preserved despite interactions with nodes outside the cycle. Therefore, for all rounds $j \geq c$, $r(j) \neq \emptyset$. This directly contradicts Definition \ref{def:termination}, which requires existence of a round $t$ such that for all $k > t$ and $\forall u \in V: M(k,u) = 0$. Since we have shown no such $t$ exists, flooding never terminates.
\qed\end{proof}

% The next corollary, which follows from Theorems \ref{thm:acyclic-termination} and \ref{thm:cyclic-nonterm}, provides a characterisation of the graphs in which RDAF terminates.

% \begin{corollary}
% \label{cor:termination-characterisation}
% A graph admits a non-terminating schedule if and only if it contains a cycle.
% \end{corollary}

\section{Periodic Infinite Schedules}
\label{sec:PIS}

We introduce \textit{Periodic Infinite Schedules (PIS)} as a framework for analysing certain non-terminating behaviours in RDAF systems, serving as a bridge between finite state descriptions and infinite executions.

\begin{definition}[Configuration]
\label{def:configuration}
A \emph{configuration} $\sigma(j)$ for round $j$ is an ordered pair $(S(j,\cdot), r(j))$ where $S$ and $r$ are the state and round functions,
% (definitions \ref{def:state-function}, \ref{def:round-function})
capturing complete system state and message transmissions.
\end{definition}

\begin{definition}[Schedule]
\label{def:schedule}
A \emph{schedule} $\sigma$ maps each round $j\in\mathbb{N}$ to its configuration.
Given a valid delay function $d$, the schedule $\sigma_d$ \emph{induced by $d$} is as follows:
\begin{enumerate}
    \item $\sigma_d(0) = (S(0, \cdot), r(0))$,
    \item For $j > 0$, $\sigma_d(j) = (S(j, \cdot), r(j))$ where:
    \begin{itemize}
        \item $S(j, v) = u(v, (S(j-1, v), d(j-1, v, \cdot)))$,
        \item $r(j) = \{(v, \{v,w\}) \in V \times E : M(j-1, v) = 1 \land w \in \text{dest}(j-1, v) \land d(j, v, \{v,w\}) = 0\}$.
    \end{itemize}
\end{enumerate}
\end{definition}

In the above definition, the schedule is ``induced'' as states and transmissions arise deterministically from applying the delay function according to our state evolution rules.

\begin{definition}[Eventually Periodic Delay Function]
\label{def:periodic-delay}
A \emph{delay function} $d$ is \textit{eventually periodic} with period $p$ if, for some $c \in \mathbb{N}$, we have that $d(i, v, \{v,u\}) = d(i+p, v, \{v,u\})$, for every round $i \geq c$ and for every node $v \in V$ and every edge $\{v,u\}$.
\end{definition}

\begin{definition}[Periodic Infinite Schedule]
\label{def:PIS}
A schedule $\sigma$ is \emph{periodic infinite} if there exist natural numbers $c$ (stabilisation round) and $l$ (cycle length) where:
\begin{enumerate}
    \item $\sigma(j) = \sigma(j+l)$ for every $j \geq c$,
    \item $r(c+k) \neq \emptyset$ for at least one $k \in \{0,1,\ldots,l-1\}$.
\end{enumerate}
\end{definition}

In the above definition, the first condition establishes repeating behaviour after the stabilisation round $c$, while the second guarantees genuine non-termination through guaranteed transmissions.
Now we establish three fundamental results characterising PIS behaviour:

\begin{theorem}[Identification of Periodic Infinite Schedules (IPIS)]
\label{thm:ipis}
Given a graph $G$ and delay function $d$, the induced schedule $\sigma_d$ is periodic infinite if:
\begin{enumerate}
    \item $d$ is eventually periodic (Definition \ref{def:periodic-delay}) with period $p$,
    \item $\exists c, l \in \mathbb{N}_{>0}: \forall u \in V: S(c, u) = S(c+l, u)$,
    \item $l \bmod p = 0$,
    \item $\exists j \in \{0,\ldots,l-1\}: r(c+j) \neq \emptyset$.
\end{enumerate}
\end{theorem}

\begin{proof}
Let $c$, $l$ satisfy condition 2. By induction on $k \geq 0$, prove $\forall j \geq c: \sigma(j) = \sigma(c + (j - c) \bmod l)$.

Base: Trivial for $k=0$ as $(c - c) \bmod l = 0$.

Step: Assume $\sigma(c + k) = \sigma(c + k \bmod l)$ for some $k \geq 0$.
For round $c + k + 1$:

1. By periodicity of $d$ and $l = bp$:
\[ d(c + k, u, e) = d(c + (k \bmod l), u, e) \]
\[ d(c + k + 1, u, e) = d(c + ((k + 1) \bmod l), u, e) \]

2. State update depends only on:
\begin{itemize}
    \item Previous states $S(c + k, \cdot)$
    \item Delay decisions $d(c + k, \cdot, \cdot)$
\end{itemize}

3. By inductive hypothesis and delay periodicity:
\[ S(c + k + 1, \cdot) = S(c + (k + 1) \bmod l, \cdot) \]
\[ r(c + k + 1) = r(c + (k + 1) \bmod l) \]

Condition 4 ensures transmission occurs in each cycle.
\qed\end{proof}

\begin{theorem}[Non-Termination of Periodic Infinite Schedules (NTPIS)]
\label{thm:ntpis}
Any periodic infinite schedule is non-terminating.
\end{theorem}

\begin{proof}
Let $\sigma$ be periodic infinite with period $c$, length $l$. Assume it terminates at $t$.
By PIS definition, $\exists i \in \{0,\ldots,l-1\}: r(c+i) \neq \emptyset$.

Let $k = c + i + l \cdot \left\lceil\frac{t}{l}\right\rceil$. Then:
\begin{itemize}
    \item $k > t$,
    \item $(k - c) \bmod l = i$.
\end{itemize}

Thus $\sigma(k) = \sigma(c + i)$ implying $r(k) = r(c+i) \neq \emptyset$.
But $k > t$ requires $r(k) = \emptyset$ and $\forall u: M(k,u) = 0$—contradiction.
\qed\end{proof}

\begin{theorem}[EPIS: Existence of PIS]
\label{thm:epis}
If a graph admits a non-terminating schedule, it admits a periodic infinite schedule.
\end{theorem}

\begin{proof}
Suppose that $d$ induces a non-terminating schedule $\sigma_d$ on $G$. The configuration space is finite, as it has at most $(2 \cdot 2^{|V|} \cdot 2^{|V|})^{|V|} \cdot 2^{2|E|}$ configurations. 
Therefore, as $\sigma_d$ is a non-terminating schedule, there exists at least one configuration $C$ which repeats infinitely often.

As the adversary respects the finite delay property, it follows that for every pair $(v,\{v,u\})$ there exists an infinite sequence of rounds, in which $(v,\{v,u\})$ is not delayed. 
Let $j_1$ be the first round in $\sigma_d$ where configuration $C$ appears. 
Due to the finite delay property, there exists some round $j_2 > j_1$ such that (i)~the configuration $C$ appears also at round $j_2$ and (ii)~every pair $(v,\{v,u\})$ was allowed to transmit (i.e.,~it was not delayed by the adversary) at least once between rounds $j_1$ and $j_2$. 

We now define a new delay function $d'$, which is periodic with period $j_2 - j_1$ after round~$j_2$, as follows:
\begin{itemize}
    \item if $j\leq j_2$ then $d'(j, u, e) = d(j, u, e)$, for every $(u,e)$,
    \item if $j> j_2$ then $d'(j, u, e) = d'(j - j_2 + j_1, u, e)$, for every $(u,e)$.
\end{itemize}

Then the schedule $\sigma_{d'}$ induced by this new delay function $d'$ is periodic after round $j_2$, with period $j_2-j_1$.
\qed\end{proof}

% \begin{remark}[On PIS Periods]
% \label{rem:pis-periods}
% \begin{enumerate}
% \item Period $n$ schedules are also periodic with period $k \cdot n$ for any $k \in \mathbb{N}_{>0}$
% \item Schedule periods may be shorter than delay periods; if the configuration repeats with period $l$ and the delay function has period $p$, the schedule identified by Theorem \ref{thm:ipis} exhibits periodicity at $\text{lcm}(l,p)$.\footnote{Assuming the delay function behaviour preserves the cyclic behaviour of the induced schedule.}
% \end{enumerate}
% \end{remark}

These results establish that PIS capture the fundamental structure of non-termination in RDAF systems. While not all non-terminating schedules are periodic, any graph admitting non-termination must also admit a periodic infinite schedule. This insight reduces termination analysis to the study of periodic behaviours, bridging finite state descriptions and infinite executions.

\section{Undecidability with Arbitrary Computable Adversaries}\label{sec-undecidability}

We prove that determining flooding termination in RDAF is undecidable when the adversary is an arbitrary computable function via reduction from the Halting problem\cite{Lucas2021}. The decision problem for the termination of RDAF is defined as follows:

\medskip
\textsc{RDAF-TERMINATION}
\begin{itemize}
    \item \textbf{Input:} A graph $G = (V,E)$, source node $g_0 \in V$, and computable delay function $d$.
    \item \textbf{Question:} Does flooding terminate on $G$ with source $g_0$ under adversary $d$?
\end{itemize}

We first establish a basic non-terminating strategy on the triangle graph $G = (V,E)$ where:

\begin{itemize}
    \item $V = \{\text{Source}, A, B\}$,
    \item $E = \{\{\text{Source}, A\}, \{A, B\}, \{B, \text{Source}\}\}$.
\end{itemize}

% \begin{wrapfigure}{r}{0.25\textwidth}
%     \centering
%     \includegraphics[width=0.25\textwidth]{Undecidability/UndecidableGraph.png}
%     \caption{A triangle graph.}
%     \label{fig:SimpleTriangle}
% \end{wrapfigure}
We now define the basic delay function $d_0$:
\[
d_0(j, u, e) = \begin{cases} 
1, & \text{if } j \bmod 3 = 0 \text{ and } e \neq \{\text{Source}, B\}, \\
1, & \text{if } j \bmod 3 = 1 \text{ and } e \neq \{\text{Source}, A\}, \\
1, & \text{if } j \bmod 3 = 2 \text{ and } e \neq \{A, B\}, \\
0, & \text{otherwise.}
\end{cases}
\]

\begin{lemma}[Basic Strategy Nonterminating]
\label{lem:basic-strategy}
The basic delay strategy $d_0$ creates a periodic infinite schedule.
\end{lemma}

\begin{proof}
    We demonstrate the basic strategy induces a periodic infinite schedule by verifying the conditions of Theorem \ref{thm:ipis}:
    \begin{enumerate}
        \item The delay function $d$ is eventually periodic with period $p = 3$. For all $i \geq 1$, $u \in V$, $e \in E$: $d(i, u, e) = d(i + 3, u, e)$. 
        This follows from the three-round delay pattern in the strategy definition.
    
        \item Taking $c = 3$ and $l = 6$, we observe from the state evolution in Table \ref{table:basic-strategy} that $S(3, u) = S(9, u)$ for all $u \in V$. 
        % The complete state evolution demonstrating this equality is shown below:
    
        \begin{table}[ht!]
        \centering
        \caption{State Evolution under Basic Strategy}
        \label{table:basic-strategy}
        \begin{tabular}{|c|c|c|c|c|}
        \hline
        \textbf{Round 0} & \textbf{Round 1} & \textbf{Round 2} & \textbf{Round 3} & \textbf{Round 4} \\
        \hline
        $S: (1, \emptyset, \{A,B\})$ & $S: (1, \emptyset, \{B\})$ & $S: (1, \emptyset, \{B\})$ & $S: (1, \{B\}, \{A\})$ & $S: (0, \emptyset, \emptyset)$ \\
        $A: (0, \emptyset, \emptyset)$ & $A: (1, \{S\}, \{B\})$ & $A: (0, \emptyset, \emptyset)$ & $A: (0, \emptyset, \emptyset)$ & $A: (1, \{S\}, \{B\})$ \\
        $B: (0, \emptyset, \emptyset)$ & $B: (0, \emptyset, \emptyset)$ & $B: (1, \{A\}, \{S\})$ & $B: (1, \{S\}, \{A\})$ & $B: (1, \emptyset, \{A\})$ \\
        \hline
        \textbf{Round 5} & \textbf{Round 6} & \textbf{Round 7} & \textbf{Round 8} & \textbf{Round 9} \\
        \hline
        $S: (0, \emptyset, \emptyset)$ & $S: (1, \{B\}, \{A\})$ & $S: (1, \{A\}, \{B\})$ & $S: (1, \emptyset, \{B\})$ & $S: (1, \{B\}, \{A\})$ \\
        $A: (1, \{B\}, \{S\})$ & $A: (1, \emptyset, \{S\})$ & $A: (1, \{S\}, \{B\})$ & $A: (0, \emptyset, \emptyset)$ & $A: (0, \emptyset, \emptyset)$ \\
        $B: (1, \{A\}, \{S\})$ & $B: (0, \emptyset, \emptyset)$ & $B: (0, \emptyset, \emptyset)$ & $B: (1, \{A\}, \{S\})$ & $B: (1, \{S\}, \{A\})$ \\
        \hline
        \end{tabular}
        \end{table}

        \item $l \bmod p = 0$ as $6 \bmod 3 = 0$.
    
        \item In each cycle $[3,9)$, transmission occurs. For instance, at round 4, we have $r(4) \neq \emptyset$.
        % \item  
    \end{enumerate}

    Therefore, by Theorem \ref{thm:ipis}, the schedule is periodic infinite with $c = 3$ and $l = 6$. By Theorem \ref{thm:ntpis}, we conclude it is non-terminating.
\qed\end{proof}

\begin{theorem}[Undecidability]
\label{thm:undecidability}
\textsc{RDAF-TERMINATION} is undecidable.
\end{theorem}

\begin{proof}
Let $M$ be an arbitrary Turing Machine $M$, and let $x$ be an arbitrary input to $M$. Then we construct the following delay function on the triangle graph $G=(V,E)$:
\vspace{-0.2cm}

\[ d(j, u, e) = \begin{cases}
0, & \text{if } M \text{ halts on } x \text{ within } j \text{  steps}, \\
d_0(j, u, e), & \text{otherwise.}
\end{cases} \]
where $d_0$ is the basic delay function defined above.
We will prove that $M$ halts on $x$ if and only if flooding terminates under the delay function $d$.

($\Rightarrow$) Suppose that $M$ halts on $x$ after exactly $t$ steps. Then $d(j, u, e) = 0$ for every $j\geq t$ and every $u$ and $e$. Then flooding terminates by round $t+2$. 

($\Leftarrow$) Suppose that flooding terminates at round $t$, and assume for the sake of contradiction that $M$ does not halt after any finite number of steps. Then $d(j,u,e)=d_0(j,u,e)$ for every $j$ and every $u$ and $e$, and thus $d$ creates a periodic infinite schedule by Lemma \ref{lem:basic-strategy}. Therefore flooding does not terminate after any finite number of rounds, which is a contradiction.
\qed\end{proof}

This undecidability persists even under severe computational constraints:

\begin{theorem}[Resilient Undecidability]
\label{thm:resilient-undecidable}
For any unbounded computable $f: \mathbb{N} \to \mathbb{N}$ with $\lim_{n \to \infty} f(n) = \infty$, an adversary with $O(f(n))$ time and space in round $n$ can simulate $\lfloor f(n) \rfloor$ Turing machine steps, preserving undecidability.
\end{theorem}

\begin{proof}
Let the adversary in round $n$ simulate $\lfloor f(n) \rfloor$ Turing machine steps, applying no delays if halted, else $d_0$. Then the Turing machine halts after $k$ steps if and only if flooding terminates after $f^{-1}(k)$ rounds. The theorem holds even for extremely slow-growing $f$ like inverse Ackermann $\alpha(n) = \min\{m : A(m,m) > n\}$ where $A$ is the Ackermann function\cite{Matos2016}.
\qed\end{proof}

This fundamental limitation motivates restricting adversary behaviour rather than computational power, leading to the EPA model.

\section{Eventually Periodic Adversaries}
\label{sec:EPA}

To bridge the gap between undecidability and practical analysis, we introduce Eventually Periodic Adversaries (EPA), which restrict adversaries to eventually periodic behaviour while maintaining significant expressive power.

\begin{definition}[Eventually Periodic Adversary]
\label{def:EPA}
An \emph{Eventually Periodic Adversary} is a triple $(d, c, l)$ where the delay function $d$ is a computable function, $c \in \mathbb{N}$ is the stabilisation round, $l \in \mathbb{N}^+$ is the cycle length, and 
$d(i,u,e) = d(i + l,u,e)$ for every $i \geq c$, $u \in V$, and $e \in E$. 
\end{definition}

\begin{lemma}[State Space Size]
\label{lemma:finite-state-space}
There exist at most $\left(2^{2|V|} + 1\right)^{|V|} \cdot 2^{|E|} \cdot (c + l)$ possible configurations.
\end{lemma}

\begin{proof}
Each configuration comprises:
\begin{enumerate}
    \item Node states $S(j,v)$ for each $v \in V$ where:
        \begin{itemize}
            \item Case $M(j,v) = 0$: Exactly one choice (empty $s(j,v), \text{dest}(j,v)$).
            \item Case $M(j,v) = 1$: $2^{|V|}$ choices for $s(j,v) \times 2^{|V|}$ choices for $\text{dest}(j,v)$.
            \item Total per node: $2^{2|V|} + 1$ choices.
            \item Total for all nodes: $\left(2^{2|V|} + 1\right)^{|V|}$ choices.
        \end{itemize}
    \item Delay decisions $d(j,v,e)$ for each edge $e \in E$: One binary choice per edge $(2^{|E|}$ choices).
    \item Round number (modulo periodicity): Values in $\{0,\ldots,c+l-1\}$ $((c+l)$ choices$)$.
\end{enumerate}

Multiplying together the total number of choices for these independent parameters implies that
$\left(2^{2|V|} + 1\right)^{|V|} \cdot 2^{|E|} \cdot (c + l)$.
\qed\end{proof}

% \subsection{Complexity}
\begin{theorem}[EPA Decidability]
\label{epa-decidable}
The termination problem for EPA-RDAF systems is decidable with time complexity $O(2^{2|V|^2 + 2|E|} (c+l))$.
\end{theorem}

\begin{proof}
We construct a decision procedure:
\begin{enumerate}
    \item Maintain set $E$ of explored states, initialised empty.
    \item Begin from initial configuration $\sigma(0)$.
    \item For each round $j$:
        \begin{enumerate}
            \item if all destination sets empty: Return ``Terminates'',
            \item if $j \geq c$ and $\sigma(j) \in E$: Return ``Non-terminating'',
            \item add $\sigma(j)$ to $E$,
            \item compute $\sigma(j+1)$ using state update rules and EPA delays.
        \end{enumerate}
\end{enumerate}

The procedure halts because:
\begin{itemize}
    \item State space is finite per Lemma \ref{lemma:finite-state-space}
    \item After round $c$, any new state must either:
        \begin{itemize}
            \item reach termination (empty destination sets by Lemma \ref{obs:PEDS}),
            \item enter a cycle (repeat a configuration).
        \end{itemize}
\end{itemize}

Complexity analysis:
\begin{itemize}
    \item explores $O(2^{2|V|^2} \cdot 2^{|E|} \cdot (c+l))$ configurations,
    \item per configuration: $O(|V| + |E| + \log(|S|))$ operations,
    \item total: $O(2^{2|V|^2 + 2|E| + \epsilon})$ assuming $c,l$ are not exponential in $|V|$.
\end{itemize}
\qed\end{proof}

\begin{theorem}[Lower Bound]
\label{EPA-RDAF-LOWER-BOUND}
Any decision procedure for EPA-RDAF termination has time complexity $\Omega((c + l) \cdot |E|)$.
\end{theorem}

\begin{proof}
Each round requires $\Theta(|E|)$ operations to process message transmissions ($\Theta(\sum_{u \in V} \text{deg}(u)) = \Theta(|E|)$) and apply delays. Minimum $(c + l)$ rounds needed to detect cycle/termination. Therefore $\Omega((c + l) \cdot |E|)$ operations required.
\qed\end{proof}

The EPA model demonstrates that restricting adversary behaviour to eventual periodicity yields decidable termination while preserving significant expressive power, bridging the gap between undecidability for arbitrary computable adversaries and practical analysis needs. A non-trivial lower bound remains open.

\section{Conclusions and Open Problems}
\label{app:conclusion}
We have established fundamental properties of Round-Delayed Amnesiac Flooding through rigorous mathematical characterisation. Our analysis yields a complete structural dichotomy: acyclic graphs guarantee termination with $O((B + 1) \cdot e(g_0))$ bound for $B$-bounded delays, while cyclic graphs admit non-terminating adversarial strategies. This extends to a sharp complexity separation - termination is undecidable for arbitrary computable adversaries, but it becomes decidable for eventually periodic adversaries with time complexity at least $\Omega((c + l)|E|)$ and at most $O(2^{2|V|^2 + 2|E|}(c + l))$.

Several theoretical challenges remain open. The complexity landscape invites tighter bounds for EPA-RDAF, while specific graph classes may admit improved parameterisation by structural properties. Natural model extensions include multiple messages and dynamic topologies with bounded modification rates. 
%while preserving decidability. 
% Beyond our undecidability result, characterising minimal constraints yielding decidability remains open, particularly for intermediate models like pushdown automata.

Our dichotomy results establish fundamental limits on automated verification of asynchronous flooding protocols: while acyclic networks allow termination analysis, verification becomes undecidable in the presence of cycles unless the adversary exhibits eventual periodicity. This framework precisely characterises when efficient algorithmic analysis of asynchronous flooding behaviour is possible.

\bibliographystyle{splncs04}
\bibliography{bibliography-algowin}

\end{document}